\documentclass[preprint,12pt]{elsarticle}
\usepackage{amssymb}
\usepackage[para,online,flushleft]{threeparttable}
\graphicspath{ {./figures/} }

\usepackage{float}
\usepackage{verbatim}
\restylefloat{figure}
\restylefloat{table}
\usepackage{pgfplots}
\usepackage{titlesec}
\titleclass{\subsubsubsection}{straight}[\subsection]
\newcounter{subsubsubsection}[subsubsection]
\renewcommand\thesubsubsubsection{\thesubsubsection.\arabic{subsubsubsection}}
\titleformat{\subsubsubsection}{\normalfont\normalsize\itshape}{\thesubsubsubsection.\space}{0em}{}
\titlespacing*{\subsubsubsection}{0pt}{2ex plus 1ex minus .2ex}{0.75ex plus .2ex}
\makeatletter
\def\toclevel@subsubsubsection{4}
\def\l@subsubsubsection{\@dottedtocline{4}{7em}{4em}}
\makeatother

\usepackage{graphicx}
\usepackage{amsmath,amssymb,amsfonts}
\usepackage{algorithmic}
\usepackage{textcomp}
\usepackage{booktabs}
\usepackage{array}
\usepackage{multirow}
\usepackage{url}
\usepackage{float}
\usepackage{pifont}
\newtheorem{theorem}{Theorem}
\newtheorem{Proposition}{Proposition}
\newtheorem{assumption}{Assumption}
\newtheorem{definition}{Definition}
\newtheorem{lemma}{Lemma}
\newproof{proof}{Proof}
\floatstyle{plaintop}
\restylefloat{table}
\usepackage{caption}
\usepackage{subcaption}

\usepackage{booktabs}

\newif\ifblackandwhite
\blackandwhitetrue
\usepackage{pdflscape}
\usepackage{colortbl}

\usepackage{booktabs}

\usepackage{adjustbox}
\usepackage{rotating}

\def\BibTeX{{\rm B\kern-.05em{\sc i\kern-.025em b}\kern-.08em
    T\kern-.1667em\lower.7ex\hbox{E}\kern-.125emX}}

\usepackage{geometry}
\begin{document}
\begin{frontmatter}

\title{ Opinion dynamics on switching networks}


\author{Amirreza Talebi$^a$}

\affiliation{organization={Department of Integrated Systems Engineering},
            addressline={The Ohio State University}, 
            city={Columbus},
            postcode={43210}, 
            state={OH},
            country={USA}}

\begin{abstract}
  We study opinion dynamics over a two-layer multiplex network with weighted, directed influence and reciprocal links. Within each layer, an agent weighs its neighbors equally, so the influence of a neighbor is inversely proportional to the agent's number of neighbors in that layer. Agents update their opinions over time to coordinate with their neighbors. However, the frequency of agents' interactions with neighbors in different network layers differs: one layer is always active, whereas the other one is active only once every $k$ time steps.
    Consequently, the multiplex network's adjacency matrices are time-varying.
    We characterize how the activation frequency of the second layer impacts the convergence of the opinion dynamics process. Without requiring agents to place positive weight on their own opinions, we prove that, for every activation period $k$, the opinions reach a consensus whose value depends on $k$, and that the asymptotic per-step convergence factor equals the $k$th root of the second-largest eigenvalue modulus of the one-period transition matrix. As $k$ grows, this factor generically converges to that of the always-active layer when this layer is connected; when it is disconnected, the time to consensus generically grows linearly in $k$. Numerical experiments on 1,000 random multiplex networks confirm these results.
\end{abstract}


\end{frontmatter}

\allowdisplaybreaks

\section{Introduction}\label{intro:1}

Social networks are instrumental in the dissemination of diverse types of information such as rumors, advertisements, and even diseases \cite{kleinberg2007cascading}. Within these networks, interactions vary significantly across different platforms; for instance, friends on Facebook might not interact on Instagram or other social networks \cite{boccaletti2014structure}. These varying connections highlight the importance of considering the multilayer nature of social networks to gain a thorough understanding \cite{de2016physics}. Multiplex networks, a specific type of multilayer network, consist of nodes that exist across multiple layers simultaneously \cite{boccaletti2014structure}.

Previous studies have extensively explored coordination games on single-layer networks, including the work of \cite{kleinberg2007cascading, fazeli2015duopoly, ghaderi2013opinion}. More recently, research has shifted towards examining coordination games on multilayer networks, with significant contributions from \cite{gomez2012evolution, hu2017opinion, wang2015evolutionary}, and see \cite[p. 70]{boccaletti2014structure} for a comprehensive review of the literature.

These coordination games are closely related to consensus problems, where agents aim to reach a collective agreement \cite{jadbabaie2003coordination, blondel2005convergence, hendrickx2005convergence, olshevsky2006convergence, olshevsky2009convergence}. Key studies in this area include works by \cite{blondel2005convergence, hendrickx2005convergence}. Additionally, some studies have introduced the concept of dynamic on and off layers within multilayer networks, as investigated by \cite{ding2017asynchronous, dong2017dynamics}.

Building on this body of work, we developed a model for multilayer coordination games incorporating the concept of dynamic layers. This model is based on the frameworks proposed by \cite{ghaderi2013opinion, ding2017asynchronous, dong2017dynamics} for single-layer networks. We further examined the convergence and convergence rate of this model, focusing on the impact of the switching frequency of these dynamic layers.

Our contributions are as follows. (i) We formulate the coordination game so that its best response is the equal-attention update across the active layers, and we prove that the transition matrix $C_k$ over one activation period (Section~\ref{sec:Convergence}) is primitive for every activation period $k\geq 1$, even though the diagonal entries of the adjacency matrices may be zero and layer~1 may be disconnected (Lemma~\ref{lem:lemma1}). Hence, the opinions reach a consensus, whose value depends on $k$ (Proposition~\ref{prop:opinions-converge}). (ii) We prove that the worst-case asymptotic per-step convergence factor equals $\lambda_2(C_k)^{1/k}$, where $\lambda_2(C_k)$ is the second-largest eigenvalue modulus (SLEM) of $C_k$ (Theorem~\ref{thm:rate}). (iii) We characterize infrequent activation: when layer~1 is connected, the SLEM of $C_k$ vanishes geometrically in $k$ and the consensus value converges to that of layer~1 alone (Proposition~\ref{prop:connected}); when layer~1 is disconnected, the SLEM of $C_k$ converges to that of a chain between the components of layer~1, so that, if this SLEM is positive, the time to consensus grows linearly in $k$ (Proposition~\ref{prop:disconnected}). (iv) We confirm these results on 1,000 random multiplex networks and identify a family of networks in which an intermediate activation period minimizes the time to consensus (Section~\ref{sec:numerical}).

The remainder of this paper is organized as follows: Section~\ref{sec:rel} provides a detailed literature review. Section~\ref{sec:network-model} introduces our model and its preliminaries. Section~\ref{sec:Convergence} presents an analytical analysis of the model's convergence to equilibrium. Section~\ref{sec:Convergence Rate} investigates the model's convergence rate, Section~\ref{sec:numerical} reports numerical studies, and Section~\ref{sec:Conclusions} concludes the paper.

\section{Related Work}\label{sec:rel}

In game-theoretic models, networked coordination games are crucial for understanding how individuals in online environments acquire information about new technologies, opinions, and rumors \cite{kleinberg2007cascading}. \cite{kleinberg2007cascading} studied the seeding problem, crucial in viral marketing, laying the groundwork for understanding these dynamics. However, they did not address the convergence rate of information evolution, which was later examined by \cite{montanari2009convergence} for single-layer networks. \cite{montanari2009convergence} found that well-connected networks exhibit slow convergence, whereas small, poorly connected networks converge more rapidly. \cite{ghaderi2013opinion} offered a different perspective, showing that complete graphs have faster convergence compared to poorly connected graphs like rings.
In multilayer networks, \cite{gomez2012evolution} examined the dynamics of an evolutionary game (the prisoner's dilemma) and found that multiplex networks improve the resilience of cooperative behaviors compared to single-layer networks. Likewise, \cite{hu2017opinion} studied opinion diffusion in multilayer networks, highlighting the influence of stubborn agents and the effect of strongly connected agent groups on opinion convergence.

Consensus problems, which involve coordinating a group of agents to reach a common agreement, have been extensively studied. \cite{olfati2004consensus} explored dynamic agent consensus in directed and undirected graphs, considering various assumptions such as constant and switching network topology, and time delays. \cite{ren2005consensus} showed that consensus can be reached in a single directed network with switching topologies if the union of the network topologies forms a spanning tree. \cite{blondel2005convergence} and \cite{hendrickx2005convergence} investigated the convergence rates of bidirectional equal neighbor models, emphasizing the sequences of stochastic matrices with positive diagonal elements.
Consensus problems frequently involve studying stochastic matrix products, with foundational contributions by \cite{wolfowitz1963products, hajnal1958weak}. \cite{xia2015products} emphasized that for consensus to be achieved, the left-sided product sequence of stochastic matrices must converge to a rank one matrix. More recently, \cite{xia2018generalized} introduced generalized Sarymsakov matrices, which guarantee that the product of compact subsets results in a rank one matrix. \cite{olshevsky2006convergence} investigated the convergence rate of averaging algorithms, finding bounds for the SLEM. \cite{talebi2024opiniondynamicssocialmultiplex} analyzed how opinions converge in multiplex networks through coordination games, considering both one-way and two-way interactions and the presence of a leader. Using graph theory and Markov chains, it shows that opinions generally converge, with leader-led networks aligning with the leader's opinion and leaderless networks reaching a mix of all opinions. The presence of leaders and one-way interactions notably speed up the convergence process. Our work is very similar to this work in modeling, notations, and convergence rate analysis with big differences in assumptions and proofs. 

To date, no well-known studies have examined the effects of layer switching frequency on the convergence rate of opinion dynamics in multiplex networks. Our work addresses this gap by investigating how switching frequency impacts the convergence rate, focusing on networks whose adjacency matrices may have zero diagonal entries, i.e., agents need not place weight on their own opinions.

\section{The multiplex network}\label{sec:network-model}

We study opinion dynamics on a two-layer multiplex network, wherein $n$ agents are present on, and interact with each other, over two different networks. As an example, each network layer could represent one form of social interaction (e.g., with family, coworkers, or online). Following \cite{talebi2024opiniondynamicssocialmultiplex}, let $\mathcal{M}=\{\mathcal{V},\mathcal{G}\}$ represent the multiplex network, where $\mathcal{V}=\{v_1, \ldots, v_n\}$ is the set of agents, with $v_i\in \mathcal{V}, i\in \mathcal{I}=\{1, \ldots, n\}$ denoting agent $i$,
and $\mathcal{G}$ is the set of weighted, directed graphs of the network layers, i.e., $\mathcal{G}=\{\mathcal{G}_{\alpha}=\{(\mathcal{V},\mathcal{E_{\alpha}},\mathcal{W_{\alpha}})| \mathcal{E_{\alpha}}\subseteq \mathcal{V} \times \mathcal{V}, \mathcal{W_{\alpha}}: \mathcal{E}_{\alpha}\rightarrow \mathbb{R}_{\geq 0}\}, \forall \alpha \in \{1, 2\}\}$. 

We denote the set of neighbors of agent $i$ in layer $\alpha$ by $\delta_{i\alpha} = \{v_j| (v_i,v_j)\in \mathcal{E}_{\alpha}, j\in \mathcal{I}\}$; when no confusion arises, we identify $v_j$ with its index $j$. As detailed shortly, an edge $(v_i, v_j)$ means that agent $i$'s opinion is influenced by agent $j$'s opinion. Therefore, $\delta_{i\alpha}$ represents all agents whose opinion influences that of agent $i$ on layer $\alpha$. We assume $n\geq 2$ and $\delta_{i\alpha}\neq\emptyset$ for all $i\in\mathcal{I}$ and $\alpha\in\{1,2\}$: an agent without any neighbor other than itself in layer $\alpha$ carries a self-loop there, i.e., $\delta_{i\alpha}=\{v_i\}$, so that it keeps its own opinion when only this layer is active.

\begin{figure}
\centering
\input{example.tikz}
\caption{An example of a two-layer multiplex network at $t=k$. Layer 1 has the bidirectional links A-B, A-C, B-C, and C-D, and agent E has only a self-loop; layer 2 has the bidirectional links C-D, C-E, and D-E, and agents A and B have only self-loops. This network satisfies Assumption~\ref{as:strongly-connected} and is used as Example 1 in Section~\ref{sec:numerical}.
\label{fig:1}
}
\end{figure}

\subsection{The adjacency matrices} 

In this paper, we study the evolution of agents' opinions as they interact with each other over an infinite, discrete time horizon. We consider scenarios in which one layer of the multiplex network may not necessarily be \emph{active} at each time step $t$. This allows us to capture situations in which agents interact with neighbors in different layers, and are impacted by their opinions, at different frequencies. For instance, an individual might interact daily with online friends, yet only infrequently with their extended family. Specifically, we assume that layer 1 is always active, while layer 2 becomes active only every $k$ time steps, i.e., at times $\{k, 2k, \ldots\}$, where $k\geq 1$ is the activation period ($k=1$ means that both layers are always active). We capture these time dependencies through the adjacency matrix of the multiplex network. Our modeling and notations are very similar to those of \cite{talebi2024opiniondynamicssocialmultiplex}.

Formally, we assume the weighted adjacency matrix $\mathcal{A}_\alpha$ of a single layer $\alpha$ 
is given by
\begin{align}\label{eqn:ad1}
[\mathcal{A}_{\alpha}]_{ij}= \begin{cases}
\frac{1}{|\delta_{i\alpha}|},& j \in \delta_{i\alpha},\\
0, & \text{otherwise,}
\end{cases}\qquad \forall i,j\in\mathcal{I},
\end{align}
where $[\mathcal{A}_{\alpha}]_{ij}$ denotes the $ij^{\textsuperscript{th}}$ element of matrix $\mathcal{A}_\alpha$. In other words, agent $i$ places equal weight on the opinion of neighbors $j$ who influence $i$'s opinion in layer $\alpha$. Note that we allow for self-loops, in which case the agent itself counts as one of its ``neighbors''. In particular, $[\mathcal{A}_{\alpha}]_{ii}=1$ if and only if $\delta_{i\alpha}=\{v_i\}$, i.e., agent $i$ has no neighbor other than itself in layer $\alpha$.
 Based on these, the weighted adjacency matrix of the multiplex network for different time steps is defined as follows $\forall i\in \mathcal{I}$: 
  \begin{equation}
    [A(t)]_i=
    \begin{cases}
      \frac{1}{2}([\mathcal{A}_1]_i+[\mathcal{A}_2]_i), &  \text{if}\ t\equiv 0\ \mod\ k\ \& \ \\
      & [\mathcal{A}_1]_{ii}<1\ \& \ [\mathcal{A}_2]_{ii}<1\\
      [\mathcal{A}_{\alpha}]_i, & \text{if}\ \ t\equiv 0\ \mod\ k\ \&\\
      & [\mathcal{A}_{\beta}]_{ii}=1,\ \{\alpha,\beta\}=\{1,2\}\\
      [\mathcal{A}_1]_i, & \text{otherwise.}
    \end{cases}\label{eqn:m2}
  \end{equation}
  Here, $[A]_i$ denotes row $i$ of matrix $A$; if $[\mathcal{A}_1]_{ii}=[\mathcal{A}_2]_{ii}=1$, both rows equal $\mathbf{e}_i^\top$ and either case of \eqref{eqn:m2} applies. Since $A(t)$ depends only on whether $t\equiv 0 \mod k$, we write $B:=A(k)$; thus, $A(t)=B$ if $t\equiv 0 \mod k$ and $A(t)=\mathcal{A}_1$ otherwise.

In other words, with this definition of the adjacency matrix of the multiplex network, when both layers are active, each agent divides its attention equally between the two layers, unless it has no neighbor other than itself in one of them, in which case it listens to the other layer only.
Note also that matrices $\mathcal{A}_1$, $\mathcal{A}_2$ and $A(t), \forall t$, are row stochastic matrices i.e., $\sum_{j}[\mathcal{A}_1]_{ij} = \sum_{j}[\mathcal{A}_2]_{ij} = \sum_{j}[A(t)]_{ij} =1, \ \forall i,t$.

\subsection{Agents' opinion dynamics}\label{sec:opinion-model}

Each agent  $i\in \mathcal{I}$ has an opinion $x_i(t) \in [0,1]$ at time $t$. We study how this opinion evolves based on interactions with other agents on the multiplex network. Let $\mathbf{x}(t) = [x_1(t), \ldots, x_n(t)]^\top$ denote the vector of agents' opinions at time $t$, and $\mathbf{x}(0)$ denote agents' initial opinions.

We consider an extension of the coordination game (proposed in \cite{ghaderi2013opinion} for single-layer networks) to a multiplex network. In this game, at each time $t+1$ agent $i$ chooses its new opinion $y$ to minimize the following cost function, given the opinions $\mathbf{x}(t)$:
\begin{equation}\label{eqn:cost}
    J_i(y;\mathbf{x}(t))=
    \begin{cases}
    \frac{1}{4}\sum_{\alpha=1}^{2}\frac{1}{|\delta_{i\alpha}|}\sum_{j\in\delta_{i\alpha}}(y-x_j(t))^2, \text{ if: }\\
    t+1\equiv 0\ \mod\ k,\ \& \ [\mathcal{A}_1]_{ii}<1\ \& \ [\mathcal{A}_2]_{ii}<1\\
    \\
    \frac{1}{2}\sum_{j\in\delta_{i\alpha}}(y-x_j(t))^2, \text{ if: }\\
    t+1\equiv 0\ \mod\ k\ \&\  [\mathcal{A}_{\beta}]_{ii}=1,\ \{\alpha,\beta\}=\{1,2\}\\
    \\
    \frac{1}{2}\sum_{j\in \delta_{i1}}(y-x_j(t))^2,\ \text{otherwise}. \\
    \end{cases}
\end{equation}

In other words, agent $i$ attempts to minimize the difference in its opinion with its neighbors; the above cost captures how this set changes depending on which layers are active at the update time $t+1$. If $v_i\in\delta_{i\alpha}$ (a self-loop), the term $(y-x_i(t))^2$ penalizes deviations from the agent's current opinion. The factors $1/|\delta_{i\alpha}|$ in the first case let the agent divide its attention equally between the two layers; without them, the minimizer would weigh each neighbor in each layer by $1/(|\delta_{i1}|+|\delta_{i2}|)$, which differs from \eqref{eqn:m2} whenever $|\delta_{i1}|\neq|\delta_{i2}|$.

Using the first-order condition $\frac{d J_i(y;\mathbf{x}(t))}{d y}=0$, the best response strategy of agent $i$ at time $t+1$ will be given by:

\begin{equation}\label{eqn:best-response}
    x_i(t+1)=
    \begin{cases}
    \frac{1}{2}\sum_{\alpha=1}^{2}\sum_{j\in\delta_{i\alpha}}\frac{x_j(t)}{|\delta_{i\alpha}|},\ if:\\
    t+1\equiv 0\mod\ k,\ \& \ [\mathcal{A}_1]_{ii}<1\ \& \ [\mathcal{A}_2]_{ii}<1\\
    \\
    \sum_{j\in\delta_{i\alpha}}\frac{x_j(t)}{|\delta_{i\alpha}|},\ if:\\
    t+1\equiv 0\ \mod\ k\ \& [\mathcal{A}_{\beta}]_{ii}=1,\ \{\alpha,\beta\}=\{1,2\}\\
    \\
    \sum_{j\in \delta_{i1}}\frac{x_j(t)}{|\delta_{i1}|}, \text{ otherwise.}\\
    \end{cases}
\end{equation}
It is easy to verify that $\frac{d^2 J_i(y;\mathbf{x}(t))}{d y^2}$ equals $1$ in the first case of \eqref{eqn:cost} and $|\delta_{i\alpha}|\geq 1$ in the other two cases; hence, $J_i$ is strictly convex in $y$, and the first-order condition yields the unique best response. Also, as $A(t+1)$ is a row stochastic matrix, under \eqref{eqn:best-response}, $0\leq \mathbf{x}(t)\leq 1$ implies $0\leq \mathbf{x}(t+1)\leq 1$.

Writing the agents' best-response dynamics \eqref{eqn:best-response} in a compact form, we have
\begin{align}\label{eqn:update-dynamics}
    \mathbf{x}(t+1) = A(t+1) \mathbf{x}(t),\quad t\geq 0,
\end{align}
so that $\mathbf{x}(t)=A(t)A(t-1)\cdots A(1)\mathbf{x}(0)$.

The update equation in \eqref{eqn:update-dynamics} is similar to those arising in the study of opinion dynamics on single-layer networks with time-varying adjacency matrices (e.g., \cite{ hendrickx2005convergence,jadbabaie2003coordination, olshevsky2009convergence,ren2005consensus}). 
The main distinction of our analysis is in our focus on the periodic switching frequency in this graph (as motivated by a multiplex network setting), which enables us to derive sharper, tailored characterizations of the opinions' convergence rate accordingly.   

For a non-negative $n\times n$ matrix $A$, we denote by $\mathcal{G}_{A}:=\{\mathcal{V}, \mathcal{E}_A\}$ its associated network, with $(v_i,v_j)\in\mathcal{E}_A$ if and only if $[A]_{ij}>0$; e.g., $\mathcal{G}_{\mathcal{A}_\alpha}$ is layer $\alpha$ (including self-loops) and $\mathcal{G}_B$ is the multiplex network when both layers are active. To proceed with our analysis, we make the following assumptions about the multiplex networks.
\begin{assumption}\label{as:strongly-connected}
\begin{enumerate}
    \item Every strongly connected component of layer 1 with at least two agents contains an odd cycle. If no such component exists (i.e., $\mathcal{A}_1=I$), $\mathcal{G}_B$ contains an odd cycle. A self-loop counts as a cycle of length one.
    \item The multiplex network when both layers are active is strongly connected, i.e., $\mathcal{G}_B$ is strongly connected.
    \item In both layers, agents' interactions are bidirectional (yet in general have different weights), i.e., $v_j\in\delta_{i\alpha}$ if and only if $v_i\in\delta_{j\alpha}$.
\end{enumerate}
\end{assumption}
We note that the presence of odd cycles (the first assumption) is essential to avoid bipartite structures, and has been common in prior work (e.g., \cite[Theorem 1.2]{bondy1976graph}, \cite[p. 8]{levin2017markov}, \cite{ghaderi2013opinion}); Section~\ref{sec:numerical} reports networks in which layer 1 and $\mathcal{G}_B$ are bipartite and the opinions do not converge. The first assumption is sufficient but not necessary, as networks in which the odd cycles are supplied by layer 2 also converge (Section~\ref{sec:numerical}).
Further, the second assumption which states that $A(t)$ is strongly connected whenever $t\equiv\ 0\ \mod\ k$ is a mild assumption, as any disconnected network components can be studied separately. By the third assumption, the strongly connected components of layer 1 are its connected components, and no link of layer 1 joins two of them.

\section{Convergence of the Opinion Profile}\label{sec:Convergence}

We begin by establishing that the opinion dynamics in \eqref{eqn:update-dynamics} will converge. The distinguishing character of the proof of convergence specific to our problem provided in this part is that by leveraging tools from matrix and graph theories, we have proved the convergence of the left product of stochastic matrices whose diagonal entries may be zero and whose layer-1 factor need not be irreducible.

We use the following definitions in this section.

\begin{definition}
\begin{enumerate}
    \item \cite{wolfowitz1963products} A (row) stochastic matrix $A$ is called \emph{SIA} (standing for Stochastic, Indecomposable, and Aperiodic) if there exists a rank one matrix $L$ such that:
    \[L = \lim_{\ell\rightarrow \infty} A^\ell~.\]
    Since $A^\ell\mathbf{1}=\mathbf{1}$, such an $L$ has the form $\mathbf{1}\pi^\top$ with $\pi\geq 0$ and $\pi^\top\mathbf{1}=1$.
    \item \cite[Theorem 8.5.2.]{horn2012matrix} A non-negative  matrix $A$ is called \emph{primitive} if there exists an integer $\ell\geq 1$ s.t. $A^\ell>0$, i.e., every element of $A^\ell$ is positive.
    \item A non-negative matrix $A$ is \emph{irreducible} if $\mathcal{G}_A$ is strongly connected. The \emph{period} of an irreducible $A$ is the greatest common divisor of the lengths of all closed walks in $\mathcal{G}_A$, and $A$ is \emph{aperiodic} if its period is one. An irreducible matrix is primitive if and only if it is aperiodic \cite[Section 8.5]{horn2012matrix}, \cite[p. 694]{meyer2000matrix}.
\end{enumerate}
\end{definition}
A primitive stochastic matrix is SIA, with $\pi>0$ \cite[Theorem 8.5.1]{horn2012matrix}. The converse is false: an SIA matrix may be reducible, e.g., $\left[\begin{smallmatrix}1&0\\1&0\end{smallmatrix}\right]$.

To analyze the convergence of agents' opinions, we begin by taking a closer look at the updates carried out over time. We note that the agents update their opinions solely on the active layer 1 for the first $k-1$ time steps, while at time $k$, they are influenced by neighbors in both layer 1 and layer 2. Therefore, denoting $C :=A(k)\mathcal{A}_1^{k-1}=B\mathcal{A}_1^{k-1}$ (we write $C_k$ when the dependence on $k$ matters), we have
\begin{align*}
&\mathbf{x}(t) = \mathcal{A}_1^t \mathbf{x}(0), ~\forall t=1, \ldots, k-1\\ 
&\mathbf{x}(k) = C \mathbf{x}(0)~.
\end{align*}
Or more compactly, and for subsequent time steps $t$, we have: 
\begin{align}
\mathbf{x}(t) &= \mathcal{A}_1^{t-k\lfloor \frac{t}{k}\rfloor} C^{\lfloor \frac{t}{k}\rfloor} \mathbf{x}(0).
\end{align}
As such, the convergence of the opinion profile $\mathbf{x}(t)$ depends on the properties of the matrix $C$, which we first establish in the following lemma.

\begin{lemma}\label{lem:lemma1}

Given Assumption~\ref{as:strongly-connected}, the matrix $C=A(k)\mathcal{A}_1^{k-1}$ is primitive for every $k\geq 1$; in particular, it is an SIA matrix.
\end{lemma}

\begin{proof}

First, we note the products of row stochastic matrices result in a row stochastic matrix \cite[8.7.P1]{horn2012matrix}; hence, $C$ is a row stochastic matrix. We next show it is also irreducible and aperiodic. Let $s:=k-1\geq 0$.

We begin by noting that the underlying graph of layer 1, $\mathcal{G}_{\mathcal{A}_1}=\{\mathcal{V}, \mathcal{E}_{\mathcal{A}_1}\}$ might \emph{not} be strongly connected. In general, this layer can consist of a number of individual nodes with only self-loops, and strongly connected components each having more than one node. Formally, define $\mathcal{V}_{\xi} := \{v_i|(v_i,v_i)\in \mathcal{E}_{\mathcal{A}_1}, (v_i,v_j)\notin \mathcal{E}_{\mathcal{A}_1},\forall j\neq i\in \mathcal{I}\}$, and $\mathcal{E}_{\xi} := \{(v_i,v_i)|v_i\in \mathcal{V}_{\xi}\}$. 

In other words, $\mathcal{V}_{\xi}$ and $\mathcal{E}_{\xi}$ correspond to the agents with only self-loops and their self-loops in layer 1, respectively; $[\mathcal{A}_1]_i=\mathbf{e}_i^\top$ for $v_i\in\mathcal{V}_\xi$. 
Additionally, let $M$ denote the set of strongly connected components $m$ of $\mathcal{G}_{\mathcal{A}_1}$ with node sets $\mathcal{V}_m$ s.t. $|\mathcal{V}_m|>1$. Let $\mathcal{A}_{1m}$ denote the adjacency matrix of the strongly connected component $m$, i.e., the principal submatrix of $\mathcal{A}_1$ indexed by $\mathcal{V}_m$. Further, let $N:= \{i|v_i\in \mathcal{V}-\mathcal{V}_{\xi}\}$ denote the set of nodes that are in one of these components. By Assumption~\ref{as:strongly-connected}.3, no link of layer 1 leaves a component; hence, every row of $\mathcal{A}_1^\ell$, $\ell\geq 0$, indexed by $\mathcal{V}_m$ is supported on $\mathcal{V}_m$, and $\mathcal{A}_{1m}$ is a stochastic matrix. Moreover, $\mathcal{A}_{1m}$ is irreducible, its graph contains cycles of length two (links are bidirectional) and an odd cycle (Assumption~\ref{as:strongly-connected}.1). Hence, the greatest common divisor of the cycle lengths of each component $m\in M$ is one, and $\mathcal{A}_{1m}$ is primitive \cite[Section 8.5]{horn2012matrix}: there exists $\gamma_m$ s.t. $\mathcal{A}_{1m}^\ell>0$ for all $\ell\geq\gamma_m$.

\textit{Step 1 (support of $B$).} Let $i\in N$. If $[\mathcal{A}_2]_{ii}<1$, then $[B]_i=\frac{1}{2}([\mathcal{A}_1]_i+[\mathcal{A}_2]_i)$; otherwise, $[B]_i=[\mathcal{A}_1]_i$, by \eqref{eqn:m2}. In both cases,
\begin{equation}\label{eqn:B-dominates}
[B]_{ij}\geq \tfrac{1}{2}[\mathcal{A}_1]_{ij},\qquad \forall i\in N,\ j\in\mathcal{I}.
\end{equation}
Next, let $u\neq v$ with $v\in\delta_{u2}$. Then $[\mathcal{A}_2]_{uu}<1$, and \eqref{eqn:m2} gives $[B]_u=[\mathcal{A}_2]_u$ if $v_u\in\mathcal{V}_\xi$ and $[B]_u=\frac{1}{2}([\mathcal{A}_1]_u+[\mathcal{A}_2]_u)$ otherwise; in both cases, $[B]_{uv}>0$. Hence, $\mathcal{E}_{\mathcal{A}_1}\setminus\mathcal{E}_\xi\subseteq\mathcal{E}_B$, and every link of layer 2 between two distinct agents belongs to $\mathcal{E}_B$. Conversely, $\mathcal{E}_B\subseteq \mathcal{E}_{\mathcal{A}_1}\cup\mathcal{E}_{\mathcal{A}_2}$.

\textit{Step 2 (components of layer 1).} Let $m\in M$ and $i\in\mathcal{V}_m$. By \eqref{eqn:B-dominates}, $[C]_i=[B]_i\mathcal{A}_1^{s}\geq\frac{1}{2}[\mathcal{A}_1]_i\mathcal{A}_1^{s}=\frac{1}{2}[\mathcal{A}_1^{s+1}]_i$ entrywise. Restricting to $\mathcal{V}_m$, the principal submatrix $C_{mm}$ of $C$ satisfies $C_{mm}\geq\frac{1}{2}\mathcal{A}_{1m}^{s+1}$; since all matrices involved are non-negative, $(C^r)_{mm}\geq (C_{mm})^r\geq 2^{-r}\mathcal{A}_{1m}^{r(s+1)}$ for every $r\geq 1$. Therefore, $(C^r)_{mm}>0$ whenever $r(s+1)\geq\gamma_m$. Consequently, $\mathcal{G}_C$ contains a walk between any two agents of $\mathcal{V}_m$, and every agent of $\mathcal{V}_m$ lies on closed walks of $\mathcal{G}_C$ of lengths $r$ and $r+1$ for all sufficiently large $r$.

\textit{Step 3 (links of layer 2).} Let $u\neq v$ with $v\in\delta_{u2}$, so that $[B]_{uv}>0$ by Step 1. As $\mathcal{A}_1^s$ is stochastic, there is $w$ with $[\mathcal{A}_1^s]_{vw}>0$, and $[C]_{uw}\geq[B]_{uv}[\mathcal{A}_1^s]_{vw}>0$. If $v_v\in\mathcal{V}_\xi$, then $[\mathcal{A}_1^s]_v=\mathbf{e}_v^\top$, so $w=v$ and $(v_u,v_v)\in\mathcal{E}_C$. Otherwise, $v_v$ and $v_w$ belong to the same component $m\in M$, and by Step 2, $\mathcal{G}_C$ contains a walk from $v_w$ to $v_v$. In both cases, $\mathcal{G}_C$ contains a walk from $v_u$ to $v_v$.

\textit{Step 4 (irreducibility).} By Step 1, every link $(v_u,v_v)$, $u\neq v$, of $\mathcal{G}_B$ is either a link of layer 1 inside some component $m\in M$ or a link of layer 2. By Steps 2 and 3, $\mathcal{G}_C$ contains a walk from $v_u$ to $v_v$ in both cases. Since $\mathcal{G}_B$ is strongly connected (Assumption~\ref{as:strongly-connected}.2), concatenating these walks along a path of $\mathcal{G}_B$ shows that $\mathcal{G}_C$ is strongly connected, i.e., $C$ is irreducible \cite{meyer2000matrix}.

\textit{Step 5 (aperiodicity).} If $N\neq\emptyset$, Step 2 provides an agent that lies on closed walks of $\mathcal{G}_C$ of lengths $r$ and $r+1$; hence, the period of $C$ divides $\gcd(r,r+1)=1$. If $N=\emptyset$, then $\mathcal{A}_1=I$ and $C=B$. In this case, every link of $\mathcal{G}_B$ between distinct agents is a link of layer 2 (Step 1), its reverse is also a link of $\mathcal{G}_B$ (Assumption~\ref{as:strongly-connected}.3 and Step 1), and such a link exists since $n\geq 2$ and $\mathcal{G}_B$ is strongly connected; thus, $\mathcal{G}_B$ contains a cycle of length two and, by Assumption~\ref{as:strongly-connected}.1, an odd cycle, so that the period of $C$ divides $\gcd(2,\text{odd})=1$. In both cases, $C$ is irreducible and aperiodic and, therefore, primitive \cite[Section 8.5]{horn2012matrix}, \cite[p. 694]{meyer2000matrix}.

Finally, a primitive stochastic matrix satisfies $\lim_{\ell\rightarrow\infty}C^\ell=\mathbf{1}\pi^\top$ with $\pi>0$ \cite[Theorem 8.5.1]{horn2012matrix}, which has rank one; hence, $C$ is SIA. Since the argument holds for every $s=k-1\geq 0$, the lemma holds for every $k\geq 1$.
 
\end{proof}

We are now ready to prove our convergence result. 

\begin{Proposition}\label{prop:opinions-converge}
Under Assumption~\ref{as:strongly-connected} on the multiplex network, for every $k\geq 1$ and every $\mathbf{x}(0)$, the opinion dynamics in \eqref{eqn:update-dynamics} converges to the consensus profile $\mathbf{\bar{x}}=(\pi_k^\top\mathbf{x}(0))\mathbf{1}_n$, i.e., $\lim_{t\rightarrow \infty} \mathbf{x}(t) = \mathbf{\bar{x}}$, where $\pi_k>0$ is the stationary distribution of $C_k$. 
\end{Proposition}

\begin{proof}

Based on lemma \ref{lem:lemma1}, we know that $C$ is an SIA matrix. From  \cite{wolfowitz1963products,ghaderi2013opinion} for SIA matrices we have
\begin{align}
\lim_{t\rightarrow \infty}C^{\lfloor \frac{t}{k}\rfloor} = \mathbf{1}_n \mathbf{\pi}^{\top}_n,
\end{align}
where $\mathbf{\pi}=\pi_k$ is a column vector of dimension $n\times 1$ indicating the unique stationary distribution of the Markov Chain with probability transition matrix $C$. Note that since $\mathcal{A}_1$ is a row stochastic matrix, we have $\mathcal{A}_1^t\mathbf{1}_n = \mathbf{1}_n$ and $||\mathcal{A}_1^t||_\infty=1, \forall t$. Writing $r_t:=t-k\lfloor \frac{t}{k}\rfloor\in\{0,\ldots,k-1\}$, we obtain $\mathbf{x}(t)-\mathbf{1}_n\mathbf{\pi}_n^{\top} \mathbf{x}(0)=\mathcal{A}_1^{r_t}(C^{\lfloor \frac{t}{k}\rfloor}-\mathbf{1}_n\mathbf{\pi}_n^{\top})\mathbf{x}(0)$, whose $\infty$-norm is at most $||(C^{\lfloor \frac{t}{k}\rfloor}-\mathbf{1}_n\mathbf{\pi}_n^{\top})\mathbf{x}(0)||_\infty\rightarrow 0$. So, 
\begin{align}\label{eqn:convergence-A(t)}
\lim_{t\rightarrow \infty} \mathbf{x}(t) &= \lim_{t\rightarrow \infty}\mathcal{A}_1^{t-(k)\lfloor \frac{t}{k}\rfloor} C^{\lfloor \frac{t}{k}\rfloor} \mathbf{x}(0)\nonumber \\
&=\lim_{t\rightarrow \infty}\mathcal{A}_1^{t-(k)\lfloor \frac{t}{k}\rfloor}\mathbf{1}_n\mathbf{\pi}_n^{\top} \mathbf{x}(0)= \mathbf{1}_n\mathbf{\pi}_n^{\top} \mathbf{x}(0) =\mathbf{\bar{x}}.
\end{align}

Thus, the opinions $\mathbf{x}(t)$ converge as $t\rightarrow \infty$. 
\end{proof}

Every consensus profile $c\mathbf{1}_n$ is a fixed point of \eqref{eqn:update-dynamics} and an equilibrium of the coordination game, since it makes every cost \eqref{eqn:cost} vanish. Proposition~\ref{prop:opinions-converge} shows that, from a given $\mathbf{x}(0)$, the dynamics select the equilibrium with $c=\pi_k^\top\mathbf{x}(0)$, a weighted average of the initial opinions in which every agent has a positive weight. The weights $\pi_k$ depend on the activation period $k$; we study this dependence in Propositions~\ref{prop:connected} and \ref{prop:disconnected}.

\section{Convergence Rate}\label{sec:Convergence Rate}

In this section, we are investigating the convergence rate of the opinion dynamics model \eqref{eqn:update-dynamics} based on  \cite{blondel2005convergence}, \cite{olshevsky2009convergence}, and \cite{talebi2024opiniondynamicssocialmultiplex}.

The matrix $C$ is a stochastic matrix. Therefore, its spectral radius is one. Since $C$ is primitive (Lemma~\ref{lem:lemma1}), the eigenvalue one is simple and every other eigenvalue of $C$ has modulus strictly less than one \cite[Theorem 8.5.1]{horn2012matrix}. For a stochastic matrix $P$, we define its SLEM $\lambda_2(P)$ as the largest modulus among the eigenvalues of $P$, counted with algebraic multiplicity, after one copy of the eigenvalue one is removed; thus, $\lambda_2(C)<1$, whereas $\lambda_2(P)=1$ if the eigenvalue one of $P$ is not simple. So, the SLEM of the matrix will bound the convergence rate of matrix $C$.
Since the matrix $C$ is not guaranteed to be symmetric, we might confront complex eigenvalues. This is not a barrier since we are dealing with the modulus of the SLEM. We denote the spectral radius of a square matrix by $\rho(\cdot)$ and use the following identity repeatedly.

\begin{lemma}\label{lem:deflation}
Let $P$ be an $n\times n$ stochastic matrix and $w\in\mathbb{R}^n$ with $w^\top\mathbf{1}=1$. Then the eigenvalues of $P-\mathbf{1}w^\top$, counted with algebraic multiplicity, are $0$ together with the eigenvalues of $P$ after one copy of the eigenvalue one is removed. In particular, $\rho(P-\mathbf{1}w^\top)=\lambda_2(P)$.
\end{lemma}
\begin{proof}
Let $T=[\mathbf{1}\ V]$ be an invertible matrix whose first column is $\mathbf{1}$. Since $P\mathbf{1}=\mathbf{1}$, we have $T^{-1}PT=\left[\begin{smallmatrix}1&a^\top\\ 0&\tilde P\end{smallmatrix}\right]$ for some vector $a$ and matrix $\tilde P$, whose eigenvalues are those of $P$ with one copy of the eigenvalue one removed. Since $T^{-1}\mathbf{1}=\mathbf{e}_1$ and $w^\top\mathbf{1}=1$, we have $T^{-1}\mathbf{1}w^\top T=\left[\begin{smallmatrix}1&w^\top V\\ 0&0\end{smallmatrix}\right]$. Hence, $T^{-1}(P-\mathbf{1}w^\top)T=\left[\begin{smallmatrix}0&a^\top-w^\top V\\ 0&\tilde P\end{smallmatrix}\right]$, which proves the claim.
\end{proof}

\begin{lemma}[\cite{blondel2005convergence}]\label{lem:lemma2}
Let $C$ be a primitive stochastic matrix with stationary distribution $\pi$, and let $\mathbf{\bar{x}}=\mathbf{1}\pi^\top\mathbf{x}(0)$. For every $q\in(\lambda_2(C),1)$, there exists $U>0$ s.t., for all integers $\ell\geq 0$ and all $\mathbf{x}(0)$, the opinion dynamics with weighted adjacency matrix $C$ satisfies
$$||C^\ell\mathbf{x}(0) - \mathbf{\bar{x}}||_{\infty}\leq 2Uq^\ell||\mathbf{x}(0)||_2.$$
\end{lemma}
\begin{proof}
Since $C\mathbf{1}\pi^\top=\mathbf{1}\pi^\top C=(\mathbf{1}\pi^\top)^2=\mathbf{1}\pi^\top$, induction on $\ell$ gives $C^\ell-\mathbf{1}\pi^\top=(C-\mathbf{1}\pi^\top)^\ell$ for $\ell\geq 1$. By Lemma~\ref{lem:deflation}, $\rho(C-\mathbf{1}\pi^\top)=\lambda_2(C)<q$; hence, Gelfand's formula \cite[Corollary 5.6.14]{horn2012matrix} implies $U':=\sup_{\ell\geq 0}||C^\ell-\mathbf{1}\pi^\top||_\infty q^{-\ell}<\infty$. The claim follows with $U:=U'/2$, since $||\mathbf{x}(0)||_\infty\leq||\mathbf{x}(0)||_2$.
\end{proof}

\begin{Proposition}\label{prop:rate-bound}

Let $q\in(\lambda_2(C_k),1)$ and $U$ be as in Lemma~\ref{lem:lemma2} with $C=C_k$. The convergence rate of the opinion dynamics \eqref{eqn:update-dynamics} is as follows \cite{talebi2024opiniondynamicssocialmultiplex}:
\begin{align}
    ||\mathbf{x}(t)-\mathbf{\bar{x}}||_{\infty} \leq 2U||\mathcal{A}_1^{t-k\lfloor \frac{t}{k}\rfloor}||_\infty q^{\lfloor \frac{t}{k}\rfloor}||\mathbf{x}(0)||_2= 2Uq^{\lfloor \frac{t}{k}\rfloor}||\mathbf{x}(0)||_2.
\end{align}

\end{Proposition}

\begin{proof}

 From \cite{olshevsky2009convergence}, \cite[p. 341]{horn2012matrix}, and Lemma~\ref{lem:lemma2} observe that:
 \begin{align}\label{eqn:cont2}
    &||\mathbf{x}(t)-\mathbf{\bar{x}}||_{\infty} = \nonumber\\
    &||A(t)...A(1)\mathbf{x}(0)-\mathbf{1}\pi^\top\mathbf{x}(0)||_{\infty}=\nonumber\\
    &||\mathcal{A}^{t-k\lfloor\frac{t}{k}\rfloor}_1 C^{\lfloor\frac{t}{k}\rfloor}\mathbf{x}(0)-\mathcal{A}^{t-k\lfloor\frac{t}{k}\rfloor}_1\mathbf{1}\pi^\top\mathbf{x}(0)||_{\infty}\leq\nonumber\\
    &||\mathcal{A}^{t-k\lfloor\frac{t}{k}\rfloor}_1||_{\infty}||C^{\lfloor\frac{t}{k}\rfloor}\mathbf{x}(0)-\mathbf{1}\pi^\top\mathbf{x}(0)||_{\infty}\leq \nonumber\\
    & 2U||\mathcal{A}_1^{t-k\lfloor \frac{t}{k}\rfloor}||_\infty q^{\lfloor \frac{t}{k}\rfloor}||\mathbf{x}(0)||_2,
 \end{align}
where we used $\mathcal{A}_1^{r}\mathbf{1}=\mathbf{1}$ and $||M\mathbf{v}||_\infty\leq||M||_\infty||\mathbf{v}||_\infty$ for the maximum row sum norm $||\cdot||_\infty$. Finally, $||\mathcal{A}_1^{r}||_\infty=1$ for every $r\geq 0$, as $\mathcal{A}_1^{r}$ is row stochastic. And the proof is done.

\end{proof}

Proposition~\ref{prop:rate-bound} bounds the error after $t$ steps by a multiple of $q^{\lfloor t/k\rfloor}$, i.e., roughly $(q^{1/k})^t$. The following theorem shows that this per-step factor is exact as $q\downarrow\lambda_2(C_k)$.

\begin{theorem}\label{thm:rate}
Let Assumption~\ref{as:strongly-connected} hold, let $k\geq 1$, and let $\rho_k:=\lambda_2(C_k)^{1/k}$. Then
\begin{equation}\label{eqn:rate-exact}
\limsup_{t\rightarrow\infty}\Big(\sup_{\mathbf{x}(0)\in[0,1]^n}||\mathbf{x}(t)-\mathbf{\bar{x}}||_\infty\Big)^{1/t}=\rho_k,
\end{equation}
and the limit superior is attained along $t=\ell k$, $\ell\rightarrow\infty$. Consequently, for every $\eta>0$, there exists $U_\eta$ s.t. $\sup_{\mathbf{x}(0)\in[0,1]^n}||\mathbf{x}(t)-\mathbf{\bar{x}}||_\infty\leq U_\eta(\rho_k+\eta)^t$ for all $t\geq 0$, whereas, if $\rho_k>0$, no bound of this form holds with a factor smaller than $\rho_k$.
\end{theorem}
\begin{proof}
Let $\Phi(t):=A(t)\cdots A(1)=\mathcal{A}_1^{r}C^{\ell}$, where $t=\ell k+r$ and $0\leq r<k$, and let $e(t):=\sup_{\mathbf{x}(0)\in[0,1]^n}||\mathbf{x}(t)-\mathbf{\bar{x}}||_\infty$. As $\mathbf{x}(t)-\mathbf{\bar{x}}=(\Phi(t)-\mathbf{1}\pi^\top)\mathbf{x}(0)$ and $(\Phi(t)-\mathbf{1}\pi^\top)\mathbf{1}=\mathbf{0}$, writing $\mathbf{x}(0)=\frac{1}{2}(\mathbf{1}+\mathbf{z})$ with $||\mathbf{z}||_\infty\leq 1$ gives $e(t)=\frac{1}{2}||\Phi(t)-\mathbf{1}\pi^\top||_\infty$. Since $\mathcal{A}_1^r\mathbf{1}=\mathbf{1}$, we have $\Phi(t)-\mathbf{1}\pi^\top=\mathcal{A}_1^r(C^\ell-\mathbf{1}\pi^\top)$, so that $e(t)\leq\frac{1}{2}||C^\ell-\mathbf{1}\pi^\top||_\infty$, with equality if $r=0$.
Upper bound: for $q\in(\lambda_2(C),1)$, the proof of Lemma~\ref{lem:lemma2} yields $e(t)\leq \frac{1}{2}U' q^{\ell}\leq \frac{1}{2}U'q^{t/k-1}$, so the left-hand side of \eqref{eqn:rate-exact} is at most $q^{1/k}$; letting $q\downarrow\lambda_2(C)$ shows that it is at most $\rho_k$. Lower bound: along $t=\ell k$, $e(\ell k)=\frac{1}{2}||(C-\mathbf{1}\pi^\top)^\ell||_\infty$, and Gelfand's formula together with Lemma~\ref{lem:deflation} gives $e(\ell k)^{1/(\ell k)}\rightarrow\lambda_2(C)^{1/k}=\rho_k$. The last statement follows from the definition of the limit superior.
\end{proof}

Theorem~\ref{thm:rate} makes precise in which sense the activation frequency determines the convergence rate: asymptotically, every activation period of length $k$ contracts the worst-case deviation from consensus by the factor $\lambda_2(C_k)$, so that about $\ln(1/\varepsilon)/(-\ln\rho_k)$ time steps reduce it by a factor $\varepsilon\in(0,1)$. Since $C_k$ depends on $k$ in a nonlinear way, $\rho_k$ need not be monotone in $k$ (Section~\ref{sec:numerical}).

\begin{Proposition}\label{eqn:prop3}
For every $\mathbf{x}(0)$, the function
$f(t) = ||\mathbf{x}(t)-\mathbf{\bar{x}}||_{\infty}$ is non-increasing. Moreover, $\max_i x_i(t)$ is non-increasing and $\min_i x_i(t)$ is non-decreasing in $t$.
\end{Proposition}
\begin{proof}
By Proposition~\ref{prop:opinions-converge}, $\mathbf{\bar{x}}=c\mathbf{1}$ with $c=\pi_k^\top\mathbf{x}(0)$. Since $A(t+1)$ is row stochastic, $A(t+1)\mathbf{\bar{x}}=\mathbf{\bar{x}}$, and \eqref{eqn:update-dynamics} gives $\mathbf{x}(t+1)-\mathbf{\bar{x}}=A(t+1)(\mathbf{x}(t)-\mathbf{\bar{x}})$. Therefore, $f(t+1)\leq||A(t+1)||_\infty f(t)=f(t)$. Similarly, each $x_i(t+1)$ is a convex combination of the entries of $\mathbf{x}(t)$, which proves the second statement.
\end{proof}
The function $f$ need not be strictly decreasing: if $f(t)$ is attained at an agent without neighbors other than itself in layer 1, whose opinion does not change while layer 2 is inactive, then $f$ remains constant until the next activation of layer 2.

We now turn to infrequent activation, i.e., large $k$, and distinguish whether layer 1 is connected. By Assumption~\ref{as:strongly-connected}.3, $[\mathcal{A}_1]_{ij}|\delta_{i1}|=1=[\mathcal{A}_1]_{ji}|\delta_{j1}|$ for every link of layer 1. Hence, $\mathcal{A}_1$ is reversible with respect to the degrees $|\delta_{i1}|$ (self-loops included), it has real eigenvalues, and, if layer 1 is connected, its stationary distribution is $[\pi_{\mathcal{A}_1}]_i=|\delta_{i1}|/\sum_{j}|\delta_{j1}|$.

\begin{Proposition}\label{prop:connected}
Let Assumption~\ref{as:strongly-connected} hold and let layer 1 be connected, i.e., $\mathcal{A}_1$ is irreducible. Then $\mathcal{A}_1$ is primitive and $\lim_{t\rightarrow\infty}\mathcal{A}_1^t=\mathbf{1}\pi_{\mathcal{A}_1}^\top$. Moreover, with $E_k:=\mathcal{A}_1^{k-1}-\mathbf{1}\pi_{\mathcal{A}_1}^\top$:
\begin{enumerate}
    \item $\lambda_2(C_k)\leq||E_k||_\infty$ and $||\pi_k-\pi_{\mathcal{A}_1}||_1\leq||E_k||_\infty$ for every $k\geq 1$;
    \item $\lim_{k\rightarrow\infty}||E_k||_\infty^{1/(k-1)}=\lambda_2(\mathcal{A}_1)<1$. Hence, $\lambda_2(C_k)\rightarrow 0$ and $\pi_k\rightarrow\pi_{\mathcal{A}_1}$ geometrically fast, and $\limsup_{k\rightarrow\infty}\rho_k\leq\lambda_2(\mathcal{A}_1)$;
    \item if the SLEM of $\mathcal{A}_1$ is attained by a single simple eigenvalue $\lambda\neq 0$ with right and left eigenvectors $v$ and $w$, $w^\top v=1$, and if $w^\top Bv\neq 0$, then $\lim_{k\rightarrow\infty}\rho_k=\lambda_2(\mathcal{A}_1)$;
    \item if $\pi_{\mathcal{A}_1}^\top B=\pi_{\mathcal{A}_1}^\top$, then $\pi_k=\pi_{\mathcal{A}_1}$ for all $k$. In general, however, $\pi_k\neq\pi_{\mathcal{A}_1}$ (Table~\ref{tab:examples}), so that
    \[\lim_{t\rightarrow \infty}\mathcal{A}_1^{t-k\lfloor \frac{t}{k}\rfloor} C^{\lfloor \frac{t}{k}\rfloor}=\mathbf{1}\pi_k^\top\neq\mathbf{1}\pi_{\mathcal{A}_1}^\top=\lim_{t\rightarrow \infty}\mathcal{A}_1^t.\]
\end{enumerate}
\end{Proposition}
\begin{proof}
Layer 1 is connected and contains an odd cycle (Assumption~\ref{as:strongly-connected}.1) and cycles of length two (Assumption~\ref{as:strongly-connected}.3); hence, $\mathcal{A}_1$ is primitive and $\mathcal{A}_1^t\rightarrow\mathbf{1}\pi_{\mathcal{A}_1}^\top$ \cite[Theorem 8.5.1]{horn2012matrix}. Since $B\mathbf{1}=\mathbf{1}$, we have $C_k-\mathbf{1}\pi_{\mathcal{A}_1}^\top=BE_k$.
(i) By Lemma~\ref{lem:deflation} with $w=\pi_{\mathcal{A}_1}$, $\lambda_2(C_k)=\rho(BE_k)\leq||B||_\infty||E_k||_\infty=||E_k||_\infty$. Moreover, as $\pi_k^\top B\mathbf{1}=1$, we have $\pi_k^\top-\pi_{\mathcal{A}_1}^\top=\pi_k^\top C_k-\pi_k^\top B\mathbf{1}\pi_{\mathcal{A}_1}^\top=(\pi_k^\top B)E_k$, where $\pi_k^\top B$ is a probability vector; the second bound follows from $||u^\top M||_1\leq||u||_1||M||_\infty$.
(ii) Let $E:=\mathcal{A}_1-\mathbf{1}\pi_{\mathcal{A}_1}^\top$. As in the proof of Lemma~\ref{lem:lemma2}, $E^{k-1}=E_k$ for $k\geq 2$, and $\rho(E)=\lambda_2(\mathcal{A}_1)<1$ by Lemma~\ref{lem:deflation} and primitivity. Gelfand's formula gives the first claim; the other claims follow from (i) and $\rho_k\leq||E_k||_\infty^{1/k}$.
(iii) Since $\lambda\neq 1$, we have $\pi_{\mathcal{A}_1}^\top v=0$ and $w^\top\mathbf{1}=0$; hence, $Ev=\lambda v$ and $w^\top E=\lambda w^\top$. All other eigenvalues of $E$ (including the eigenvalue $0$ of Lemma~\ref{lem:deflation}) have modulus smaller than $|\lambda|$, so the Jordan decomposition of $E$ gives $\lambda^{-(k-1)}E_k\rightarrow vw^\top$ and $\lambda^{-(k-1)}BE_k\rightarrow Bvw^\top$. The latter matrix has rank one and spectral radius $|w^\top Bv|>0$. By continuity of the spectral radius, $\lambda_2(C_k)/|\lambda|^{k-1}\rightarrow|w^\top Bv|$; therefore, $\rho_k=\lambda_2(C_k)^{1/k}\rightarrow|\lambda|=\lambda_2(\mathcal{A}_1)$.
(iv) If $\pi_{\mathcal{A}_1}^\top B=\pi_{\mathcal{A}_1}^\top$, then $\pi_{\mathcal{A}_1}^\top C_k=\pi_{\mathcal{A}_1}^\top\mathcal{A}_1^{k-1}=\pi_{\mathcal{A}_1}^\top$, and $\pi_k=\pi_{\mathcal{A}_1}$ by the uniqueness of the stationary distribution of the primitive matrix $C_k$. Table~\ref{tab:examples} reports a network for which $\pi_k\neq\pi_{\mathcal{A}_1}$.
\end{proof}

When layer 1 is disconnected, $\mathcal{A}_1^{k-1}$ no longer converges to a rank-one matrix, and infrequent activation slows consensus down.

\begin{Proposition}\label{prop:disconnected}
Let Assumption~\ref{as:strongly-connected} hold and let layer 1 have $c\geq 2$ connected components with node sets $\mathcal{V}^{(1)},\ldots,\mathcal{V}^{(c)}$, where every agent of $\mathcal{V}_\xi$ forms a singleton component. Let $R\in\mathbb{R}^{n\times c}$ have the indicator vectors $\mathbf{1}_{\mathcal{V}^{(a)}}$ as columns, and let $S\in\mathbb{R}^{c\times n}$ have the rows $\sigma_a^\top$, where $\sigma_a$ is the stationary distribution of the principal submatrix of $\mathcal{A}_1$ indexed by $\mathcal{V}^{(a)}$, extended by zeros. Define $\hat C:=SBR\in\mathbb{R}^{c\times c}$, i.e., $[\hat C]_{ab}=\sigma_a^\top B\mathbf{1}_{\mathcal{V}^{(b)}}$ is the probability that one step of $B$ moves a walker distributed according to $\sigma_a$ into component $b$. Then:
\begin{enumerate}
    \item $\lim_{k\rightarrow\infty}C_k=C_\infty:=BRS$;
    \item $\hat C$ is a primitive stochastic matrix, the eigenvalues of $C_\infty$ are those of $\hat C$ together with $n-c$ zeros, and $\lim_{k\rightarrow\infty}\lambda_2(C_k)=\lambda_2(\hat C)<1$;
    \item if $\lambda_2(\hat C)>0$, then $\rho_k\rightarrow 1$ and $k(1-\rho_k)\rightarrow-\ln\lambda_2(\hat C)$ as $k\rightarrow\infty$;
    \item $\lim_{k\rightarrow\infty}\pi_k=\pi_\infty:=S^\top\hat\pi$, where $\hat\pi$ is the stationary distribution of $\hat C$.
\end{enumerate}
\end{Proposition}
\begin{proof}
(i) Each diagonal block of $\mathcal{A}_1$ is either primitive (components with at least two agents; see the proof of Lemma~\ref{lem:lemma1}) or equal to $[1]$. Hence, $\mathcal{A}_1^{s}\rightarrow RS$ as $s\rightarrow\infty$ \cite[Theorem 8.5.1]{horn2012matrix}, and $C_k=B\mathcal{A}_1^{k-1}\rightarrow BRS$.
(ii) $\hat C\geq 0$ and $\hat C\mathbf{1}=SBR\mathbf{1}=SB\mathbf{1}=S\mathbf{1}=\mathbf{1}$. Since $\sigma_a>0$ on $\mathcal{V}^{(a)}$, $[\hat C]_{ab}>0$ if and only if $\mathcal{G}_B$ has a link from $\mathcal{V}^{(a)}$ to $\mathcal{V}^{(b)}$; as $\mathcal{G}_B$ is strongly connected, so is $\mathcal{G}_{\hat C}$. If some component $\mathcal{V}^{(a)}$ has at least two agents, \eqref{eqn:B-dominates} gives $[\hat C]_{aa}\geq\frac{1}{2}\sigma_a^\top\mathcal{A}_1\mathbf{1}_{\mathcal{V}^{(a)}}=\frac{1}{2}$, so the irreducible matrix $\hat C$ has a positive diagonal entry and is primitive; otherwise, $\mathcal{A}_1=I$, $R=S=I$, and $\hat C=B=C_1$ is primitive by Lemma~\ref{lem:lemma1}. The matrices $(BR)S$ and $S(BR)$ have the same nonzero eigenvalues with the same algebraic multiplicities \cite[Theorem 1.3.22]{horn2012matrix}, which gives the spectrum of $C_\infty$. Consequently, one is a simple eigenvalue of $C_\infty$ and $\lambda_2(C_\infty)=\lambda_2(\hat C)<1$. The eigenvalues of $C_k$ converge to those of $C_\infty$ \cite[Appendix D]{horn2012matrix}; since every $C_k$ has the eigenvalue one and the other eigenvalues of $C_\infty$ are bounded away from one, the remaining $n-1$ eigenvalues of $C_k$ converge to those of $C_\infty$ other than one, i.e., $\lambda_2(C_k)\rightarrow\lambda_2(\hat C)$.
(iii) By (ii), $k\ln\rho_k=\ln\lambda_2(C_k)\rightarrow\ln\lambda_2(\hat C)\in(-\infty,0)$. Hence, $\ln\rho_k\rightarrow 0$, i.e., $\rho_k\rightarrow 1$, and $k(1-\rho_k)=-k\ln\rho_k\,(1+o(1))\rightarrow-\ln\lambda_2(\hat C)$.
(iv) The vector $\pi_\infty$ is non-negative and satisfies $\pi_\infty^\top\mathbf{1}=\hat\pi^\top S\mathbf{1}=1$ and $\pi_\infty^\top C_\infty=\hat\pi^\top SBRS=\hat\pi^\top\hat C S=\pi_\infty^\top$. Since one is a simple eigenvalue of $C_\infty$, the linear system $\pi^\top(C-I)=\mathbf{0}^\top$, $\pi^\top\mathbf{1}=1$ has full rank at $C=C_\infty$; hence, its unique solution depends continuously on $C$ near $C_\infty$, and $\pi_k\rightarrow\pi_\infty$.
\end{proof}

By Theorem~\ref{thm:rate} and Proposition~\ref{prop:disconnected}, when layer 1 is disconnected and $\lambda_2(\hat C)>0$, reducing the worst-case deviation from consensus by a factor $\varepsilon$ takes about $k\ln(1/\varepsilon)/(-\ln\lambda_2(\hat C))$ time steps for large $k$, i.e., a number of steps that grows linearly in $k$. In contrast, by Proposition~\ref{prop:connected}, when layer 1 is connected, the per-step factor $\rho_k$ stays bounded away from one as $k$ grows. The limit weights $\pi_\infty$ average the initial opinions within each component $\mathcal{V}^{(a)}$ according to $\sigma_a$ and across components according to the stationary distribution of $\hat C$, the chain in which a walker can leave a component of layer 1 only through a link of layer 2.

\section{Numerical studies}\label{sec:numerical}

We complement the analysis with numerical experiments that (a) verify the claims of Sections~\ref{sec:network-model}-\ref{sec:Convergence Rate} on random multiplex networks, including networks that violate Assumption~\ref{as:strongly-connected}, and (b) quantify how the activation period $k$ affects the convergence rate and the consensus value. The code (Python with NumPy) and all generated data accompany the manuscript; all random instances derive from a fixed master seed, so that every number reported below can be reproduced.

\subsection{Setup}\label{sec:num-setup}

Each instance is a multiplex network with $n=60$ agents drawn from one of five families, with 200 instances per family (1,000 in total). The families combine Erd\H{o}s-R\'enyi (ER) random graphs \cite{erdos1959random} with other structures:
\begin{itemize}
    \item F1 (slow layer 1, fast layer 2): layer 1 is a Watts-Strogatz small-world graph \cite{watts1998collective}, i.e., a ring in which every agent is linked to its two nearest neighbors on each side and every link is rewired with probability $0.05$, conditioned on being connected; layer 2 is an ER graph with mean degree $4.8$.
    \item F2 (fast layer 1, slow layer 2): layer 1 is an ER graph with mean degree $9$, conditioned on being connected; layer 2 is a single cycle through all agents in random order, which is bipartite since $n$ is even.
    \item F3 (disconnected layer 1): layer 1 consists of four communities of 14 agents, each an ER graph with link probability $0.4$, and four agents without links in layer 1; layer 2 is an ER graph with mean degree $1.8$. Layer 1 has 8.06 connected components on average.
    \item F4 (Assumption~\ref{as:strongly-connected}.1 violated): the agents are split into two halves, and both layers only link agents of different halves (with probabilities $0.15$ in layer 1 and $0.05$ in layer 2), so that layer 1 and $\mathcal{G}_B$ are bipartite.
    \item F5 (Assumption~\ref{as:strongly-connected}.1 violated, odd cycles supplied by layer 2): as F4, but layer 2 links any two agents with probability $0.05$.
\end{itemize}
In every family, an agent without neighbors in a layer receives a self-loop there, and no other self-loop is created. Consequently, the diagonal of $B$ is zero for all 60 agents of every instance, i.e., no agent places weight on its own opinion when both layers are active. Instances of F1-F3 are redrawn until Assumption~\ref{as:strongly-connected} holds, and instances of F4 and F5 until layer 1 is connected and $\mathcal{G}_B$ is strongly connected.

For every instance and $k\in\{1,\ldots,40\}$, we test primitivity and the SIA property of $C_k$ exactly on $\mathcal{G}_{C_k}$ (strongly connected components, closed classes, and greatest common divisor of cycle lengths, in integer arithmetic), and we compute $\lambda_2(C_k)$, $\rho_k$, $\pi_k$, and the Dobrushin coefficient $\tau_1(P):=\frac{1}{2}\max_{i,j}||[P]_i-[P]_j||_1$ \cite{seneta2006non}. For a stochastic matrix $P$, every eigenvalue $\lambda\neq 1$ satisfies $|\lambda|\leq\tau_1(P)$, and $\tau_1(P)$ is the largest ratio $\mathrm{sp}(P\mathbf{x})/\mathrm{sp}(\mathbf{x})$ over non-constant $\mathbf{x}$, where $\mathrm{sp}(\mathbf{x}):=\max_i x_i-\min_i x_i$ \cite{seneta2006non}. Hence, with $\Phi_k(t):=A(t)\cdots A(1)$,
\begin{equation}\label{eqn:teps}
T_\varepsilon(k):=\min\{t:\ \tau_1(\Phi_k(t))\leq\varepsilon\}
\end{equation}
is the worst-case number of steps after which the spread of the opinions is at most $\varepsilon$ times the initial spread; we use $\varepsilon=10^{-3}$. When layer 1 is connected, we evaluate $\lambda_2(C_k)=\rho(BE_k)$ (proof of Proposition~\ref{prop:connected}) through the symmetric matrix $D^{-1/2}(D\mathcal{A}_1)D^{-1/2}$, $D=\mathrm{diag}(|\delta_{i1}|)$, which avoids the cancellation in $C_k-\mathbf{1}\pi_k^\top$ when $\lambda_2(C_k)$ is close to machine precision. In addition, we perform the following simulations for F1-F3:
\begin{enumerate}
    \item Best response: for 50 instances per family and $k\in\{1,3,7\}$, we simulate $3k+2$ steps in which every agent minimizes its cost \eqref{eqn:cost} exactly (the minimizer of the quadratic $J_i$ is recovered from three evaluations of $J_i$), and we compare the result with \eqref{eqn:update-dynamics}.
    \item Trajectories: for every instance and $k\in\{1,2,5,10,20,40\}$, we simulate \eqref{eqn:update-dynamics} from 20 initial profiles drawn uniformly from $[0,1]^n$ for $T=\max\{500,\lceil 16\ln 10/(-\ln\rho_k)\rceil\}$ steps, so that $\rho_k^T\leq 10^{-16}$, record $\max_i|x_i(T)-\pi_k^\top\mathbf{x}(0)|$, and count violations of Proposition~\ref{eqn:prop3} (tolerance $10^{-10}$).
    \item Rate: we estimate the per-step factor by $\hat\rho_k:=(e(\ell_2k)/e(\ell_1k))^{1/((\ell_2-\ell_1)k)}$, where $e(t)=\frac{1}{2}||\Phi_k(t)-\mathbf{1}\pi_k^\top||_\infty$ (proof of Theorem~\ref{thm:rate}), $\ell_2$ is the first $\ell$ with $e(\ell k)<5\times10^{-12}$, and $\ell_1=\lceil\ell_2/2\rceil$. We report the relative error $|\ln\hat\rho_k-\ln\rho_k|/|\ln\rho_k|$ for all pairs with $\ell_2\geq 3$ (for F2 and large $k$, $e(k)$ is already close to the floor).
    \item Large $k$: we compute $\rho_k$ for $k\in\{50,100,200,400\}$.
\end{enumerate}

\subsection{Verification of the analysis}\label{sec:num-verification}

Table~\ref{tab:verification} summarizes the checks for F1-F3. In all 24,000 instance-period pairs, $C_k$ is primitive, as stated by Lemma~\ref{lem:lemma1}, although the diagonal of $B$ is zero and, in F3, layer 1 is disconnected; moreover, every entry of $\pi_k$ is positive. The best responses to \eqref{eqn:cost} coincide with \eqref{eqn:update-dynamics} up to rounding errors, whereas the cost without the factors $1/|\delta_{i\alpha}|$ yields best responses that deviate from \eqref{eqn:m2} by up to 0.27. All trajectories converge to $\pi_k^\top\mathbf{x}(0)$ (Proposition~\ref{prop:opinions-converge}), and Proposition~\ref{eqn:prop3}, the inequality $\lambda_2(C_k)\leq\tau_1(C_k)$, and both bounds of Proposition~\ref{prop:connected}(i) are never violated. The estimated per-step factors agree with $\rho_k$ (Theorem~\ref{thm:rate}): the median relative error is at most $10^{-3}$ in every family.

The large-$k$ results agree with Propositions~\ref{prop:connected} and \ref{prop:disconnected}. In F1 and F2, $\rho_k\leq||E_k||_\infty^{1/k}$ in all instances, and the median of $|\rho_k-\lambda_2(\mathcal{A}_1)|$ decreases from $1.4\times10^{-2}$ (F1) and $8.1\times10^{-3}$ (F2) at $k=50$ to $1.7\times10^{-3}$ and $1.0\times10^{-3}$ at $k=400$, i.e., $\rho_k\rightarrow\lambda_2(\mathcal{A}_1)$ as in Proposition~\ref{prop:connected}(iii). In F3, $\lambda_2(\hat C)>0$ in all instances (minimum 0.53), and $k(1-\rho_k)$ approaches $-\ln\lambda_2(\hat C)$, with median relative error $3.6\times10^{-3}$ at $k=50$ and $4.5\times10^{-4}$ at $k=400$, i.e., decreasing like $1/k$.

For the bipartite family F4, none of the 8,000 matrices $C_k$ is SIA, and the opinions do not converge: starting from the indicator vector of one half of the agents, the opinions alternate between the indicator vectors of the two halves in all 8,000 instance-period pairs. Hence, Assumption~\ref{as:strongly-connected}.1 cannot simply be dropped. For F5, in contrast, all 8,000 matrices $C_k$ are primitive although layer 1 is bipartite, which shows that Assumption~\ref{as:strongly-connected}.1 is sufficient but not necessary.

Finally, we repeated the study with $n=40$ and $n=120$ agents (100 instances per family each, with the same mean degrees). All exact checks have the same outcome (all $C_k$ primitive in F1-F3 and F5, none SIA in F4), no inequality is violated, and the dependence on $k$ described below is unchanged.

\begin{table}[t]
\centering
\caption{Verification of the analysis for $n=60$ (200 instances per family; $k\in\{1,\ldots,40\}$ unless stated otherwise). Deviations are maxima over all runs: 150 best-response runs and 24,000 trajectories per family. The last row compares $k(1-\rho_k)$ with its limit $-\ln\lambda_2(\hat C)$ (Proposition~\ref{prop:disconnected}). Empty cells: the check does not apply to the family.}
\label{tab:verification}
\footnotesize
\begin{tabular}{lccc}
\toprule
Check & F1 & F2 & F3 \\
\midrule
$C_k$ primitive (Lemma~\ref{lem:lemma1}) & 8,000/8,000 & 8,000/8,000 & 8,000/8,000 \\
$\min_{i,k}[\pi_k]_i$ & $4.9\times10^{-3}$ & $1.4\times10^{-3}$ & $3.9\times10^{-4}$ \\
Best response to \eqref{eqn:cost} vs.\ \eqref{eqn:update-dynamics} & $6.7\times10^{-16}$ & $6.7\times10^{-16}$ & $7.8\times10^{-16}$ \\
Cost without $1/|\delta_{i\alpha}|$ vs.\ \eqref{eqn:m2} & 0.23 & 0.23 & 0.27 \\
$|x_i(T)-\pi_k^\top\mathbf{x}(0)|$ & $3.4\times10^{-15}$ & $2.9\times10^{-15}$ & $8.0\times10^{-14}$ \\
Violations of Proposition~\ref{eqn:prop3} & 0 & 0 & 0 \\
Violations of $\lambda_2(C_k)\leq\tau_1(C_k)$ & 0 & 0 & 0 \\
Violations of Proposition~\ref{prop:connected}(i) & 0 & 0 & \\
Theorem~\ref{thm:rate}: median rel.\ error of $\hat\rho_k$ & $3.6\times10^{-5}$ & $9.9\times10^{-4}$ & $6.0\times10^{-5}$ \\
\quad 95th percentile (pairs used) & $5.4\times10^{-3}$ (8,000) & $1.5\times10^{-2}$ (4,581) & $2.6\times10^{-2}$ (8,000) \\
$k=400$: median $|\rho_k-\lambda_2(\mathcal{A}_1)|$ & $1.7\times10^{-3}$ & $1.0\times10^{-3}$ & \\
$k=400$: median rel.\ error of $k(1-\rho_k)$ & & & $4.5\times10^{-4}$ \\
\bottomrule
\end{tabular}
\end{table}

\subsection{Effect of the activation period}\label{sec:num-k}

Figure~\ref{fig:rate} shows $\rho_k$ as a function of $k$, and Figure~\ref{fig:teps} shows the worst-case time to consensus $T_\varepsilon(k)$. The three families behave differently, as predicted by Section~\ref{sec:Convergence Rate}.

In F1, the fast layer 2 speeds up the slow layer 1: the median of $\rho_k$ increases monotonically from 0.647 at $k=1$ to 0.960 at $k=40$ and approaches $\lambda_2(\mathcal{A}_1)$ (median 0.976), and $\rho_{40}>\rho_1$ in every instance. Accordingly, the median of $T_\varepsilon(k)$ increases from 18 to 184.5 steps, but it saturates, since $\rho_k$ stays bounded away from one (Proposition~\ref{prop:connected}).

In F2, the effect of $k$ is not monotone. Frequent activation of the slow cycle layer slows the dynamics down: $\rho_1$ (median 0.642) exceeds $\lambda_2(\mathcal{A}_1)$ (median 0.583). The median of $\rho_k$ drops to 0.480 at $k=2$ and then increases toward $\lambda_2(\mathcal{A}_1)$, with small oscillations between odd and even $k$, and $\rho_{40}<\rho_1$ in every instance. The median of $T_\varepsilon(k)$ is smallest (12 steps) at $k\in\{2,3\}$, compared with 18 steps at $k=1$ and 15 steps at $k=40$. Hence, an intermediate activation period minimizes the time to consensus.

In F3, layer 1 is disconnected, and $\rho_k$ approaches one (median 0.991 at $k=40$), following the prediction $\lambda_2(\hat C)^{1/k}$ of Proposition~\ref{prop:disconnected}. The median of $T_\varepsilon(k)$ grows linearly, from 35.5 steps at $k=1$ to 840 steps at $k=40$. A least-squares fit over $k\geq 10$ gives a median slope of 21.0 steps per unit of $k$, compared with the prediction $\ln(1/\varepsilon)/(-\ln\lambda_2(\hat C))$ (median 19.1); the ratio of the two has median 1.08 (5th and 95th percentiles 1.03 and 1.17). The slope exceeds the prediction slightly because the number of activation periods needed is an integer and also depends on the prefactor of $\lambda_2(\hat C)^{\ell}$, which the asymptotic factor does not capture.

\begin{figure}[t]
\centering
\includegraphics[width=\textwidth]{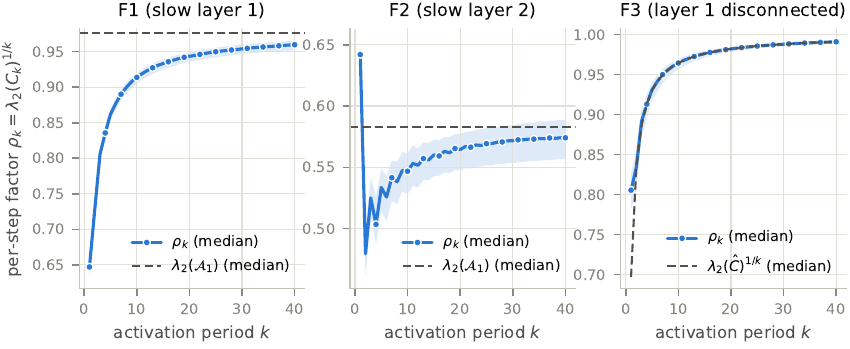}
\caption{Per-step convergence factor $\rho_k=\lambda_2(C_k)^{1/k}$ as a function of the activation period $k$ (median and interquartile range over 200 instances per family). Dashed: the median of $\lambda_2(\mathcal{A}_1)$ (F1, F2; Proposition~\ref{prop:connected}) and of $\lambda_2(\hat C)^{1/k}$ (F3; Proposition~\ref{prop:disconnected}).}
\label{fig:rate}
\end{figure}

\begin{figure}[t]
\centering
\begin{minipage}[t]{0.49\textwidth}
\centering
\includegraphics[width=\textwidth]{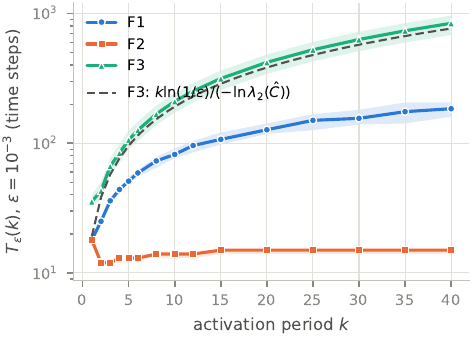}
\caption{Worst-case time to consensus $T_\varepsilon(k)$ in \eqref{eqn:teps} for $\varepsilon=10^{-3}$ (median and interquartile range). Dashed: the large-$k$ prediction for F3.}
\label{fig:teps}
\end{minipage}\hfill
\begin{minipage}[t]{0.49\textwidth}
\centering
\includegraphics[width=\textwidth]{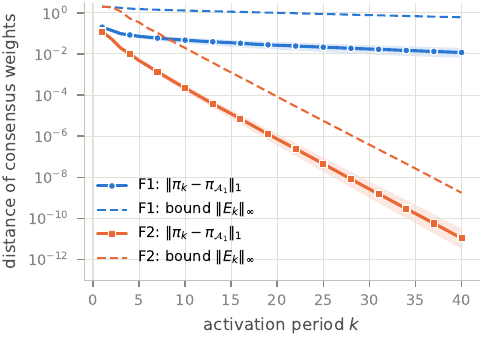}
\caption{Distance $||\pi_k-\pi_{\mathcal{A}_1}||_1$ between the consensus weights with and without layer 2 (solid; median and interquartile range, values below $10^{-14}$ are shown at $10^{-14}$) and the bound $||E_k||_\infty$ of Proposition~\ref{prop:connected} (dashed; median).}
\label{fig:consensus}
\end{minipage}
\end{figure}

\subsection{Consensus value}\label{sec:num-consensus}

The activation period also changes the consensus value. In F1 and F2, $\pi_k\neq\pi_{\mathcal{A}_1}$ in every instance: the minimum of $||\pi_k-\pi_{\mathcal{A}_1}||_1$ over all instances and $k\in\{2,\ldots,10\}$ is $7.5\times10^{-3}$ for F1 and $4.4\times10^{-5}$ for F2. As Figure~\ref{fig:consensus} shows, $||\pi_k-\pi_{\mathcal{A}_1}||_1$ decreases geometrically in $k$, below the bound $||E_k||_\infty$, and faster in F2, whose layer 1 has a smaller SLEM. Hence, the equality $\lim_{t\rightarrow\infty}\mathcal{A}_1^{t-k\lfloor t/k\rfloor}C^{\lfloor t/k\rfloor}=\lim_{t\rightarrow\infty}\mathcal{A}_1^t$ holds only asymptotically in $k$ (Proposition~\ref{prop:connected}(iv)). In F3, $\pi_k$ approaches $\pi_\infty$ of Proposition~\ref{prop:disconnected}(iv); the median of $||\pi_{40}-\pi_\infty||_1$ is $3.1\times10^{-8}$.

Table~\ref{tab:examples} illustrates these effects on the network of Figure~\ref{fig:1} (Example 1, layer 1 disconnected) and on the same network with an additional layer-1 link D-E (Example 2, layer 1 connected), with $\mathbf{x}(0)=[1,0.8,0.5,0.2,0]^\top$ for agents A-E. In Example 1, $\lambda_2(\hat C)=0.125$, so $\rho_k\rightarrow 1$, and the consensus value converges to $\pi_\infty^\top\mathbf{x}(0)=0.589$ with $\pi_\infty=[2/9,2/9,1/3,1/9,1/9]^\top$; Figure~\ref{fig:example} shows the corresponding trajectories. In Example 2, $\pi_{\mathcal{A}_1}=[0.2,0.2,0.3,0.2,0.1]^\top$ is proportional to the degrees in layer 1, $\pi_{\mathcal{A}_1}^\top\mathbf{x}(0)=0.55$, and $\rho_k$ approaches $\lambda_2(\mathcal{A}_1)=0.857$. In Example 1, $\rho_k$ is not monotone in $k$ for small $k$.

\begin{table}[t]
\centering
\caption{SLEM $\lambda_2(C_k)$, per-step factor $\rho_k$, and consensus value $\pi_k^\top\mathbf{x}(0)$ for the network of Figure~\ref{fig:1} (Example 1) and for the same network with the additional layer-1 link D-E (Example 2), with $\mathbf{x}(0)=[1,0.8,0.5,0.2,0]^\top$.}
\label{tab:examples}
\small
\begin{tabular}{rccccccc}
\toprule
 & \multicolumn{3}{c}{Example 1 (layer 1 disconnected)} & \multicolumn{4}{c}{Example 2 (layer 1 connected)}\\
\cmidrule(lr){2-4}\cmidrule(lr){5-8}
$k$ & $\lambda_2(C_k)$ & $\rho_k$ & $\pi_k^\top\mathbf{x}(0)$ & $\lambda_2(C_k)$ & $\rho_k$ & $\pi_k^\top\mathbf{x}(0)$ & $||\pi_k-\pi_{\mathcal{A}_1}||_1$\\
\midrule
1 & 0.627 & 0.627 & 0.455 & 0.663 & 0.663 & 0.384 & 0.410 \\
2 & 0.507 & 0.712 & 0.559 & 0.530 & 0.728 & 0.516 & 0.310 \\
3 & 0.314 & 0.680 & 0.586 & 0.484 & 0.785 & 0.505 & 0.116 \\
4 & 0.266 & 0.718 & 0.586 & 0.402 & 0.796 & 0.539 & 0.155 \\
5 & 0.164 & 0.696 & 0.589 & 0.354 & 0.812 & 0.532 & 0.073 \\
10 & 0.129 & 0.815 & 0.589 & 0.162 & 0.833 & 0.551 & 0.045 \\
20 & 0.125 & 0.901 & 0.589 & 0.034 & 0.845 & 0.550 & 0.008 \\
50 & 0.125 & 0.959 & 0.589 & $3.3\times10^{-4}$ & 0.852 & 0.550 & $7.8\times10^{-5}$ \\
\midrule
$k\rightarrow\infty$ & 0.125 & $1$ & 0.589 & $0$ & 0.857 & 0.550 & $0$ \\
\bottomrule
\end{tabular}
\end{table}

\begin{figure}[t]
\centering
\includegraphics[width=\textwidth]{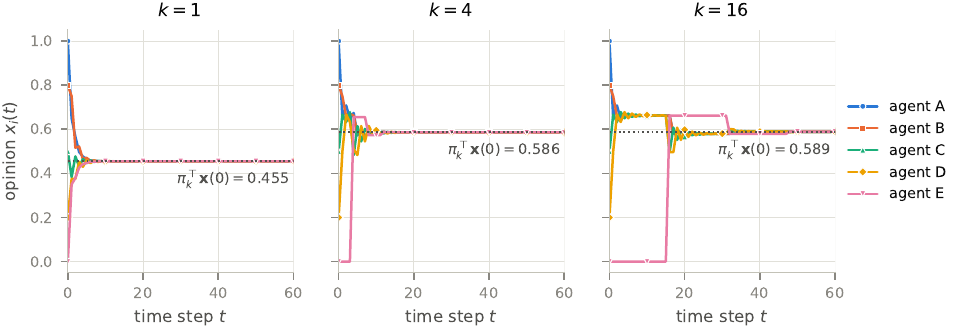}
\caption{Opinions in Example 1 (Figure~\ref{fig:1}) for $k\in\{1,4,16\}$. Dotted: the consensus value $\pi_k^\top\mathbf{x}(0)$, which depends on $k$.}
\label{fig:example}
\end{figure}

\section{Conclusions}\label{sec:Conclusions}

Including the multilayer essence of social networks is important in analyzing opinion dynamics. In this work, we modeled a coordination game on a two-layer multiplex network where one of the layers is active in all time steps, and the other one becomes activated every time after being off for $k-1$ consecutive time steps. A prevalent assumption regarding the positiveness of diagonal elements of adjacency matrices of layers was relaxed. We showed that under the mentioned conditions, the one-period transition matrix is primitive for every $k$, and opinions converge to a consensus, i.e., to the equilibrium of the coordination game that is selected by the initial opinions through the weights $\pi_k$; these weights depend on $k$. Additionally, we formulated the convergence rate of the coordination game as a function of layers switching frequency leveraging graph theory and matrix theory tools: the asymptotic per-step factor equals $\lambda_2(C_k)^{1/k}$. When layer 1 is connected, this factor generically converges to the SLEM of layer 1 as activation becomes rare; when layer 1 is disconnected, it generically tends to one and the time to consensus grows linearly in $k$. The numerical studies confirmed these results and showed that the effect of $k$ is not monotone in general: in one of the studied families of networks, an intermediate activation period minimizes the time to consensus.
To mention two venues out of many for future work, one should study the same problem in the presence of stubborn agents. As mentioned in the introduction, the location of these agents could drastically affect the opinions dynamics evolution.
Another venue is to develop the model for general multilayer networks where assumption \ref{as:strongly-connected} is relaxed, or to extend the ideas to the leader-follower model and to aperiodic or random activation of the layers.

\bibliography{references}

\end{document}